\documentclass[11pt,a4paper]{article}
\usepackage[latin1]{inputenc}

\usepackage{amsmath}
\usepackage{amsfonts}
\usepackage{amssymb}
\usepackage{graphicx}
\usepackage{amsthm}
\usepackage{enumerate}
\usepackage{setspace}
\usepackage{booktabs}
\usepackage{tabularx}
\usepackage{lineno} 
\usepackage{color} 
\usepackage{enumerate}
\usepackage{bbm}
\usepackage{xcolor}
\usepackage[title]{appendix}

\usepackage{thmtools}
\usepackage{thm-restate}
\usepackage{hyperref}
\usepackage{cleveref}

\usepackage{tikz}
\newcommand{\bE}{\ensuremath{\mathbb{E}}}
\usetikzlibrary{arrows,shapes,snakes,automata,backgrounds,petri}

\usepackage[textsize=footnotesize,color=green!40]{todonotes}

\newcounter{i}
\usepackage[top=1in, bottom=1in, left=1in, right=1in]{geometry}
\newtheorem{theorem}{Theorem}[section]

\newtheorem{lemma}[theorem]{Lemma}

\newtheorem{observation}[theorem]{Observation}

\newtheorem{conjecture}[theorem]{Conjecture}
\newtheorem{question}[theorem]{Question}

    {\noindent \emph{Proof.} {}{#1}{}}{\hfill
    $\Diamond$\vspace{1em}}
    
\theoremstyle{plain} 
\newcommand{\thistheoremname}{}
\newtheorem{genericthm}[section]{\thistheoremname}

\newcommand{\const}{84000}

\newcommand{\cA}{\mathcal{A}}

\newcommand{\cK}{\mathcal{K}}
\newcommand{\cS}{\mathcal{S}}

\newcommand{\cT}{\mathcal{T}}

\newcommand{\bP}{\mathbb{P}}
\newcommand{\diam}{\mathrm{diam}}
\newcommand{\Glau}{\mathrm{Glau}}
\newcommand{\flip}{\mathrm{flip}}

\title{Rapid mixing of Glauber dynamics for colorings below Vigoda's $11/6$ threshold}
\author{
Michelle Delcourt
\thanks{School of Mathematics,
University of Birmingham, Birmingham, UK  {\tt m.delcourt@bham.ac.uk,}. Research supported by supported by EPSRC grant EP/P009913/1.}
\and
Guillem Perarnau
\thanks{School of Mathematics,
University of Birmingham, Birmingham, UK.  {\tt g.perarnau@bham.ac.uk}.}
\and
Luke Postle
\thanks{Combinatorics and Optimization Department,
University of Waterloo, Waterloo, Ontario N2L 3G1, Canada {\tt lpostle@uwaterloo.ca}. Partially supported by NSERC
under Discovery Grant No. 2014-06162.}}
\date{\today}

\begin{document}
\maketitle
\begin{abstract}
A well-known conjecture in computer science and statistical physics is that Glauber dynamics on the set of $k$-colorings of a graph $G$ on $n$ vertices with maximum degree $\Delta$ is rapidly mixing for $k \geq \Delta +2$. In FOCS 1999, Vigoda showed rapid mixing of flip dynamics with certain flip parameters on the set of proper $k$-colorings for $k > \frac{11}{6}\Delta$, implying rapid mixing for Glauber dynamics.  In this paper, we obtain the first improvement beyond the $\frac{11}{6}\Delta$ barrier for general graphs by showing rapid mixing for $k > (\frac{11}{6} - \eta)\Delta$ for some positive constant $\eta$. The key to our proof is combining path coupling with a new kind of metric that incorporates a count of the ``extremal configurations'' of the chain. 
Additionally, our results extend to list coloring, a widely studied generalization of coloring.  Combined, these results answer two open questions from Frieze and Vigoda's 2007 survey paper on Glauber dynamics for colorings.
\end{abstract}

\newpage

\section{Introduction}

Let $G=(V,E)$ be a graph on $n$ vertices with maximum degree $\Delta$, and let $[k]$ denote the set $\left\{1,2,\ldots, k \right\}$. A \emph{$k$-coloring} of $G$ is an assignment $\sigma: V(G) \rightarrow [k]$; we say that a $k$-coloring is \emph{proper} if no two adjacent vertices receive the same color. Counting the number of proper $k$-colorings of a graph is a computationally hard problem~\cite{V79}. Jerrum, Valiant, and Vazirani~\cite{JVV86} showed that a nearly uniform sampler gives rise to an approximate enumeration, motivating the question of finding an algorithm to efficiently generate uniformly random proper colorings of a graph. This question is a central topic in computer science and statistical physics.

To this end, we study the following Markov chain Monte Carlo algorithm known as Glauber dynamics~(e.g. see~\cite{FV07}). Let $\Omega_0$ be the set of proper $k$-colorings of $G$.
The \emph{Glauber dynamics for $k$-colorings} is a discrete-time Markov chain $(X_t)$ with state space $\Omega_0$ and transitions between states given by recoloring at most one vertex; if $X_t=\sigma$, then we proceed as follows.
\begin{enumerate}
\setlength\itemsep{0em}
\item Choose $u$ uniformly at random from $V(G)$.
\item For all vertices $v \neq u$, let $X_{t+1}(v) = \sigma(v)$.
\item Choose $c$ uniformly at random from $[k]$, if $c$ does not appear among the colors in the neighborhood of $u$ then let $X_{t+1}(u) = c$, otherwise let $X_{t+1}(u) = \sigma(u)$.
\end{enumerate}

If $k \geq \Delta +1$, the greedy algorithm shows that $\Omega_0\neq \emptyset$. It is easy to check that Glauber dynamics is ergodic provided that $k\geq \Delta+2$. This is not the case for $k=\Delta+1$ as the chain might not be irreducible due to the existence of \emph{frozen colorings}, fixed points of Glauber dynamics.

A central conjecture in the area is that Glauber dynamics mixes in polynomial time (\emph{rapid mixing}) for $k\geq \Delta+2$. If so, this provides the existence of a fully polynomial almost uniform sampler (FPAUS), and by~\cite{JVV86}, an FPRAS for $k$-colorings of graphs with maximum degree $\Delta$, provided that $k\ge \Delta + 2$. A stronger version of the conjecture states that it mixes in time $O(n\log{n})$, which is best possible due to the result of Hayes and Sinclair~\cite{HS}.

Jerrum~\cite{J95} showed that for $k > 2 \Delta$ the mixing time is $O(n \log n)$. Salas and Sokal~\cite{SS} used Dobrushin's uniqueness criterion to obtain the same bound.  In 1999, Vigoda~\cite{V00} made a major breakthrough in the area by showing that an alternative chain to sample colorings, known as flip dynamics~(see below for an informal definition and see Subection~\ref{sec:flip} for a formal one), has mixing time $O(n \log n)$ for $k > \frac{11}{6}\Delta$. As a corollary and for the same range of $k$, one obtains that Glauber dynamics for $k$-colorings has mixing time $O(n^2 \log n)$ and that the $k$-state zero temperature anti-ferromagnetic Potts model on $\mathbb{Z}^d$ lies in the disordered phase when $k > \frac{11}{3}d$. It has been observed that Vigoda's result actually implies $O(n^2)$ mixing time for Glauber dynamics~(see e.g.  Chapter 14 in~\cite{LP}).

This conjecture has raised a lot of interest, and in the last 20 years, Vigoda's bound has been improved for particular classes of graphs such as graphs with large girth~\cite{DF,DFHV04,HV03,HV05,M}, trees~\cite{MSW07}, planar graphs~\cite{HVV15} and random graphs~\cite{EHSV18,MS10}. We refer the interested reader to the introduction of~\cite{EHSV18} for an extensive survey on the topic.
As to the original conjecture, no improvement over $\frac{11}{6}\Delta$ had appeared. In 2007, Frieze and Vigoda~\cite{FV07} asked whether Vigoda's approach could be pushed further.  The main contribution of this paper is to answer this question in the affirmative, thus breaking the $\frac{11}{6}\Delta$ barrier for general graphs.

\begin{theorem}\label{thm:main_Glau}
The Glauber dynamics for $k$-colorings on a graph on $n$ vertices with maximum degree $\Delta$ and $k \geq \left(\frac{11}{6}-\eta\right)\Delta$, with $\eta = \frac{1}{\const}$, has mixing time
$$
t_{\textrm{Glau}}= O\left(\left(k \log{k}\right)\cdot n^2\log{n}\right)\;.
$$
\end{theorem}

As in~\cite{V00}, Theorem~\ref{thm:main_Glau}  will follow as a corollary of a similar result for  flip dynamics, which we now introduce. 

The \emph{flip dynamics for $k$-colorings with flip parameters $\mathbf{p}=(p_1,p_2,\dots)$} is a Markov chain with space state $\Omega_0$ and transitions between states given by swapping the colors of a maximum bicolored connected set of vertices $S$ (called \emph{Kempe component}) with probability proportional to $p_{|S|}$ (see Section~\ref{sec:flip} for a precise definition). We say that $\mathbf{p}$ is \emph{bounded} if there exists an integer $\ell_0$ (independent of $n$) such that $p_\ell=0$ for every $\ell\geq \ell_0$.

\begin{theorem}\label{thm:main}
There exists a bounded $\mathbf{p}$ such that flip dynamics for $k$-colorings with flip parameters $\mathbf{p}$ on a graph on $n$ vertices with maximum degree $\Delta$ and $k \geq \left(\frac{11}{6}-\eta\right)\Delta$, with $\eta = \frac{1}{\const}$, has mixing time
$$
t_{\textrm{flip}(\mathbf{p})}\leq kn\log{(4n)}\;.
$$
\end{theorem}

Theorem~\ref{thm:main} and Lemma 7 of~\cite{V00} imply that the $k$-state zero temperature anti-ferromagnetic Potts model on $\mathbb{Z}^d$ lies in the disordered phase when $k \geq  (\frac{11}{3}-2\eta)d$.

The second part of the paper is devoted to studying the sampling of list colorings, a natural and much-studied generalization of coloring.  
Frieze and Vigoda~\cite{FV07} asked if the results obtained for sampling colorings can be transferred to list coloring. Jerrum's proof for $k> 2\Delta$~\cite{J95} carries over immediately for list coloring; however, rapid mixing of Glauber dynamics for list coloring was not previously known for $k < 2 \Delta$.  We use a modified version of flip dynamics for list colorings introduced in Section~\ref{sec:list} to obtain the analogue of Theorem~\ref{thm:main_Glau} for list colorings.

\begin{theorem}\label{thm:list_Glau}
The Glauber dynamics for $k$-list-colorings on a graph on $n$ vertices with maximum degree $\Delta$ and $k \geq \left(\frac{11}{6}-\eta\right)\Delta$, with $\eta=\frac{1}{\const}$, has mixing time $O((k\log{k})\cdot n^2\log{n})$.
\end{theorem}

To prove Theorems~\ref{thm:main_Glau} and ~\ref{thm:list_Glau}, we not only modify the flip parameters of Vigoda but we also use a different metric for path coupling.  As in most previous approaches to the study of the mixing time for Glauber dynamics, Vigoda's proof uses the Hamming metric on colorings. For our results, we use a modification of the Hamming metric. While alternative metrics have been used in the literature~\cite{BDK, BD98, LV}, those metrics have usually involved some expected ``stopping time'' and hence tend to be complicated to analyze.

Breaking from this past approach, we introduce a new kind of metric $d$: namely, the Hamming metric $d_H$ minus a small factor $d_B$ counting the number of ``non-extremal configurations'' around a vertex. 
The idea then is to prove that in just \emph{one} transition of the chain, either $d_H$ tends to decrease or $d_B$ tends to increase; in either case, this leads to a decrease in our new metric $d$. This ``extremal'' metric proves to be relatively easy to analyze, and hence we believe this concept will have fruitful applications for bounding the mixing times of other Markov chains.

As a final note, in order to prove Theorem~\ref{thm:list_Glau}, the list coloring analogue of Theorem~\ref{thm:main_Glau}, we introduce a notion of flip dynamics for list colorings wherein a Kempe component is flipped only if both colors appear in all lists of vertices of the component. While such a notion seems perfectly natural in hindsight, we are not aware of any such version in the literature.

\subsection{Structure of the paper}

Section~\ref{sec:main} is devoted to introducing notions and techniques needed for our proofs. In Section~\ref{sec:flip} we define the notion of a Kempe component and use this to describe flip dynamics. We review needed basic definitions and the path coupling technique in Section~\ref{sec:mixing}.

Our novel metric is constructed in Section~\ref{sec:metric}, where we also define extremal configurations. Informally speaking, given two colorings that differ only at a vertex $v$, we assign to each color $c\in [k]$ a configuration depending on the structure of the Kempe component containing $v$ and having color $c$. For a choice of flip parameters, some of these configurations become extremal; that is, if they are used as transitions of flip dynamics, the expected change on the distance between the coupled walks tends to increase. As an example, for the flip parameters used in~\cite{V00}, there are $6$ extremal configurations, up to symmetries.

We describe the coupling between adjacent colorings introduced in~\cite{V00} in Section~\ref{sec:coup} and we use it Section~\ref{sec:key} to state Theorem~\ref{thm:improvement}, which bounds the expected change of the path coupling in one step. In Section~\ref{sec:thm} we provide the proofs of our main results Theorems~\ref{thm:main_Glau} and~\ref{thm:main} assuming Theorem~\ref{thm:improvement}, which easily follow using path coupling.

Section~\ref{sec:change} is devoted to the proof of Theorem~\ref{thm:improvement}. As the metric $d$ is composed of two parts, we bound the contribution of the Hamming metric $d_H$ in Section~\ref{sec:nablaH} and the contribution of $d_B$ in Section~\ref{sec:nablaB} individually. 

The idea of Section~\ref{sec:nablaH} is to construct an LP to find an optimal set of flip parameters whose inequalities correspond to configurations. It turns out that there actually infinitely many optimal flip parameters for this program and that only two of the inequalities in this program are tight for every optimal solution. That is, there are always at least two extremal configurations. Thus we construct a second LP to minimize the remaining inequalities; this reduces the number of extremal configurations from $6$ to $2$ , which greatly simplifies the analysis of the path coupling in Section~\ref{sec:nablaB}.

For bounding the contribution of $d_B$ in Section~\ref{sec:nablaB}, the intuition is that if $N(v)$ contains many extremal configurations, then some of these are likely to flip to being non-extremal. More precisely, we show that the change from extremal to non-extremal is at least some constant factor of the change from non-extremal to extremal, wherein we carefully lower bound the first part using the structure of extremal configurations while upper bounding the second part with a more worst-case style of analysis.

Finally, the proof of Theorem~\ref{thm:list_Glau} is presented in Section~\ref{sec:list}, and we stress which parts are different from the non-list case and omit the ones that are analogous. We conclude with a number of open problems related to our work in Section~\ref{sec:open}.

\section{Flip dynamics, metric, and coupling}\label{sec:main}

In this section, we introduce flip dynamics, the metric $d$, and the coupling that we will use to prove our main result.

\subsection{Definition of flip dynamics}\label{sec:flip}

Let $\Omega:=[k]^n$ be the set of all colorings of $G$. Given $\sigma\in \Omega$, a path $(w_0,\dots, w_r)$ is \emph{$(c_1,c_2)$-alternating} if $\sigma(w_{j})= c_1$ for $j$ even and $\sigma(w_{j})= c_2$ for $j$ odd. 
A \emph{Kempe component} of $\sigma$ is a triplet $(c_1,c_2,S)$ where $c_1,c_2\in [k]$, and $S$ is a maximal non-empty subset of $V(G)$ such that for every $u,v\in S$ there exists a $(c_1,c_2)$-alternating path between $u$ and $v$. We slightly abuse notation and often identify the Kempe component $(c_1,c_2,S)$ with the set $S$. 
This definition is valid for proper and improper colorings; if $\sigma\in \Omega_0$, then Kempe components are maximal connected bicolored subgraphs. 
Also note that if $(c,c,S)$ is a Kempe component, then $S$ is the set of vertices of a maximal monochromatic connected subgraph.
We define the multiset
$$
\cK_\sigma:=\{S:\, (c_1,c_2,S) \text{ is a Kempe component}\}\;.
$$
Here it should be stressed that some components in $\cK_\sigma$ are taken with multiplicity. Namely, for each color $c$ that does not appear in the neighborhood of $u$, there exists a component $S=\{u\}$ in $\cK_\sigma$.  
Following the notation used in~\cite{V00}, we use $S_\sigma(u,c)$ to refer to the set $S$ in the Kempe component $(\sigma(u),c,S)$ with $u\in S$. Note that for a component $S$, there are exactly $|S|$ choices of  $(u,c)\in V(G)\times [k]$ such that $S=S_\sigma(u,c)$. It follows that $\sum_{S\in \cK_\sigma} |S| = kn$.

We say that a coloring $\sigma'$ is obtained from $\sigma$ by \emph{flipping} $S_\sigma(u,c)\in \cK_\sigma$ if
$$
\sigma'(v)=
\begin{cases}
c & \text{if }v\in S_\sigma(u,c) \text{ and }\sigma(v)=\sigma(u),\\
\sigma(u) & \text{if }v\in S_\sigma(u,c)  \text{ and }\sigma(v)=c,\\
\sigma(v) & \text{if }v\notin S_\sigma(u,c).
\end{cases}
$$
We denote by $\sigma_S$ the coloring obtained from $\sigma$ by flipping $S$.
Note that if $S=S_\sigma(u,c)$ with $\sigma(u)=c$, then $\sigma_S=\sigma$.

As described by Vigoda in~\cite{V00}, \emph{flip dynamics for $k$-colorings with flip parameters $\mathbf{p}=(p_1,p_2\dots)$} is a discrete-time Markov chain $(Y_t)$ with state space $\Omega$ and transitions between states given by swapping colors in Kempe components. In particular, if $Y_t=\sigma$, then \emph{flip dynamics} proceeds as follows.
\begin{enumerate}
\setlength\itemsep{0em}
\item Choose $u$ uniformly at random from $V(G)$.
\item Choose $c$ uniformly at random from $[k]$.
\item Let $S=S_{\sigma}(u,c)$ and $\ell= |S|$. With probability $p_\ell/\ell$, let $Y_{t+1}= \sigma_{S}$, 
otherwise, let $Y_{t+1}=\sigma$.
\end{enumerate}
For every $S\in \cK_\sigma\cup\{\emptyset\}$, define
$$
\bP_\sigma(S):=\bP(Y_{t+1}=\sigma_S\mid Y_t=\sigma)\;.
$$
If $(c_1,c_2,S)$ is a Kempe component of $\sigma$, then for each  $w\in S$ there exists a unique $c\in\{c_1,c_2\}$ such that $S=S_{\sigma}(w,c)$. Thus, if  $|S|=\ell$, then $\bP_\sigma(S)=p_\ell/kn$ and $\bP_\sigma(\emptyset)=1-\sum_{S\in \cK_\sigma} \bP_{\sigma}(S)$. 

The choice of $\mathbf{p}$ is crucial for the mixing properties of the chain. The flip parameters used by Vigoda in~\cite{V00} will be discussed in more detail in Lemma~\ref{lem:vig} and the values used in this paper can be found in Observation~\ref{obs}. 
If $p_\ell=\ell$ for every $\ell\geq 1$, flip dynamics can be understood as the Wang-Swendsen-Koteck\'{y} (WSK) algorithm for the anti-ferromagnetic Potts model at zero-temperature. The convergence properties of the WSK algorithm have received a lot of attention in the literature~\cite{LV05,WSK89,WSK90}. Throughout this paper, we will assume that $p_1=1$ and that  $p_{\ell+1}\leq p_{\ell}$, for every $\ell\geq 1$.

As a final remark, observe that while we defined flip dynamics over $\Omega_0$ in the introduction, here we define it over $\Omega$. The necessity of extending the chain to $\Omega$ will become apparent in the next section. Note that flipping a Kempe component in a proper coloring always produces a proper coloring; that is, $\Omega_0$ is a closed set of flip dynamics. Since $p_1>0$, flip dynamics embeds Glauber dynamics and thus it is ergodic on $\Omega_0$ for every $k\geq \Delta+2$. 
As every improper coloring has a positive probability to be eventually transformed into a proper one, $\Omega_0$ is the only closed subset of $\Omega$.
It follows that $(Y_t)$ converges to the uniform distribution on $\Omega_0$, denoted by $\pi$. Thus, an upper bound on the mixing time of $(Y_t)$ defined in $\Omega$ gives an upper bound on the mixing time of flip dynamics over $\Omega_0$.

\subsection{Mixing time and path coupling}\label{sec:mixing}

In this section we define the mixing time of a chain and describe the path coupling technique to obtain upper bounds on it. For any two probability distributions $\mu$ and $\nu$ supported on $\Omega$, we define its \emph{total variation distance} as
$$
\|\mu-\nu\|_{TV}=\max_{A\subseteq \Omega} |\mu(A)-\nu(A)|\;.
$$
Let $P$ be the transition matrix of flip dynamics and recall that $\pi$ is the unique stationary distribution of $(Y_t)$, which is uniform on $\Omega_0$. Define
$$
f(t)=\max_{\sigma \in\Omega} \|P^t(\sigma,\cdot)-\pi\|_{TV}\;,
$$
and
$$
t_{\text{mix}}(\epsilon)= \min \{t:\, f(t)\leq \epsilon\}\;.
$$
The \emph{mixing time} of the chain is defined as $t_{\text{mix}}:=t_{\text{mix}}\left(\frac{1}{4}\right)$.  We will denote the mixing time of Glauber dynamics by $t_{\Glau}$ and the mixing time of flip dynamics with flip parameters $\mathbf{p}$ by $t_{\flip (\textbf{p})}$.

For $i\in\{1,2\}$, let $S_i$ be a random variable over $\cK_i$ with probability distribution $\bP_i$. A \emph{coupling} of $S_1$ and $S_2$ is a joint random variable $(S_1,S_2)$ over $\cK_1\times \cK_2$ with probability distribution $ \bP_{12}$ whose marginal laws are the ones of $S_1$ and $S_2$, respectively; that is,
\begin{align*}
\sum_{S'\in \cK_2} \bP_{12}(S,S')=\bP_1(S)\;,\\
\sum_{S\in \cK_1} \bP_{12}(S,S')=\bP_2(S')\;.
\end{align*}

\medskip

A \emph{pre-metric} on $\Omega$ is a pair $(\Gamma,\omega)$ where $\Gamma$ is a connected, undirected graph with vertex set $\Omega$ and $\omega$ is a function that assigns positive, real-valued weights to edges such that for every $\sigma\tau \in E(\Gamma)$,  $\omega(\sigma \tau)$ is the minimum weight among all paths between $\sigma$ and $\tau$. 

Let $\boldsymbol{\varphi}=(\varphi_0,\varphi_1,\dots, \varphi_s)$ be a (simple) path in $\Gamma$. 
For any $\sigma',\tau'\in \Omega$, let $P_{\sigma', \tau'}$ denote the set of paths $\boldsymbol{\varphi}$ such that $\varphi_0=\sigma'$ and $\varphi_s=\tau'$.  
Let $d$ be the metric on $\Omega$ obtained by extending the pre-metric $(\Gamma,\omega)$ as follows: for every $\sigma',\tau'\in \Omega$, 
$$
d(\sigma',\tau') := \min_{\boldsymbol{\varphi}\in P_{\sigma',\tau'}} \sum_{i=1}^s \omega(\varphi_{i-1} \varphi_{i})\;.
$$

Path coupling was introduced by Bubley and Dyer~\cite{BD} to bound the mixing time of Markov chains. Here we will use the following version (see Lemma 3 in~\cite{FHY18}).
\begin{theorem}\label{thm:path coupling}[Bubley and Dyer~\cite{BD}]
Let $(\Gamma,\omega)$ be a pre-metric on $\Omega$ where $\omega$ takes values in $(0,1]$. 
Let $d$ be the metric obtained from $(\Gamma,\omega)$. 
Given $(Y_t)$ and $(Z_t)$ two copies of a chain with $Y_t Z_t\in E(\Gamma)$, let $(Y_{t+1},Z_{t+1})$ be a coupling of one step of the chain.
If there exists $\alpha>0$ such that for every $\sigma\tau\in E(\Gamma)$ one has
\begin{eqnarray}\label{premet}
\bE\left[d(Y_{t+1}, Z_{t+1})|Y_t=\sigma, Z_t=\tau \right] \leq (1-\alpha)\cdot  d(\sigma, \tau),
\end{eqnarray}
then
$$
t_{\text{mix}}\leq \frac{\log(4\,\diam(\Gamma))}{\alpha}\;,
$$
where $\diam(\Gamma)= \max_{\sigma',\tau'\in\Omega} d(\sigma',\tau')$.

\end{theorem}

\subsection{A pre-metric for the set of colorings}\label{sec:metric}

In this section we define the pre-metric we will use. Let $\Gamma$ be the graph with vertex set $\Omega = [k]^n$ where two colorings are adjacent if and only if they differ at exactly one vertex. In particular, $\diam(\Gamma)=n$. Unless otherwise stated, $\sigma$ and $\tau$ will be $k$-colorings that differ in exactly one vertex (i.e. $\sigma\tau\in E(\Gamma)$); we will always denote this vertex by $v$.  If the use of $\sigma$ and $\tau$ is interchangeable, we will often use $\varphi\in \{\sigma,\tau\}$ and $\pi\in \{\sigma,\tau\}\setminus \{\varphi\}$.

Fix an arbitrary ordering $\prec$ of $V(G)$.  Given $c \in [k]$,  let  
$$
W=\left\{w_1, w_2, \ldots, w_r\right\}:= N(v) \cap \varphi^{-1}(c)\;,
$$
 with $w_1\prec w_2 \prec \dots \prec w_r$. 
We say $(\sigma,\tau)$ has an $r$-\emph{configuration $(a_1, a_2, \ldots, a_r; b_1, b_2, \ldots, b_r)$ for} $c$ if $|S_{\tau}(w_i, \sigma(v))| = a_i$ and $|S_{\sigma}(w_i, \tau(v))| = b_i$ for all $i \in [r]$.  As in~\cite{V00}, in order to avoid the multiplicity of a Kempe component $S\in \cK_\tau$ produced by containing multiple vertices of $W$, if $W\cap S=\{w_{i_1},\dots, w_{i_j}\}$ with $w_{i_1}\prec \dots \prec w_{i_j}$, we set $a_{i_1}=|S_\tau(w_{i_1},\sigma(v))|$ and $a_{i_2}=\dots=a_{i_j}=0$, and similarly for $\cK_\sigma$ and $b_i$. 
For the sake of convenience, we consider $p_0=0$.

In order to define $\omega$ we introduce the notion of an extremal configuration. The configurations $(2;1)$ and $(1;2)$ are called \emph{extremal $1$-configurations}, and the configurations $(3,3;1,1)$ and $(1,1;3,3)$ are called \emph{extremal $2$-configurations}.  We will see in Section~\ref{sec:change} why these configurations are of particular interest.

Define the following sets of colors, 
\begin{align*}
B^1_{\sigma, \tau}(v) &:= \left\{c \in [k]: (\sigma,\tau)\text{ has an extremal 1-configuration for }c   \right\},\\
B^2_{\sigma, \tau}(v) &:= \left\{c \in [k]:  (\sigma,\tau)\text{ has an extremal 2-configuration for }c  \right\}.
\end{align*}
Let $B_{\sigma,\tau}(v)= B^1_{\sigma, \tau}(v) \cup B^2_{\sigma, \tau}(v)$ be the set of colors $c$ such that $(\sigma,\tau)$ has an extremal configuration for $c$, and $\beta_{\sigma,\tau}(v)= (|B^1_{\sigma, \tau}(v)| +2 |B^2_{\sigma, \tau}(v)|)/\Delta$ be the proportion of neighbors of $v$ that participate in extremal configurations of $(\sigma,\tau)$.

Let $\gamma\in \left(0,\frac{1}{2}\right)$ be a sufficiently small constant to be fixed later. We define
\begin{align}\label{eq:omega}
\omega(\sigma, \tau) := 1 - \gamma(1- \beta_{\sigma,\tau}(v))\;.
\end{align}

Note that $\omega(\sigma,\tau)\in [1-\gamma,1]$. Since $\gamma<\frac{1}{2}$ and $\beta_{\sigma,\tau}(v)\leq 1$, every path containing at least two edges has weight greater than one. So every edge is a minimum weight path, implying that $(\Gamma,\omega)$ is a pre-metric. Let $d$ be the metric of $\Omega$ obtained from $(\Gamma,\omega)$. 

Let $d_H$ be the Hamming metric on $\Omega$; that is, for any $\sigma',\tau'\in \Omega$,
$$
d_H(\sigma',\tau'):=|\{u\in V(G): \sigma'(u)\neq \tau'(u)\}|\;.
$$
Since the underlying state space is $\Omega=[k]^n$, for every $\sigma',\tau'\in \Omega$ there exists a path between $\sigma'$ and $\tau'$ in $\Gamma$ of length $d_H(\sigma',\tau')$ in which every edge has weight at most $1$. It follows that $d(\sigma',\tau')\leq d_H(\sigma',\tau')$.

Define
\begin{align}\label{eq:d_B}
d_B(\sigma',\tau') := d_H (\sigma', \tau')- d(\sigma', \tau')\;.
\end{align}
In general, $d_B$ is not a metric, here we will only use that it is non-negative. The contribution of $d_B$ will be crucial for the constant improvement over $\frac{11}{6}$.

\subsection{A coupling of flip dynamics for adjacent states}\label{sec:coup}

In order to use Theorem~\ref{thm:path coupling}, we need to define a coupling of flip dynamics for adjacent states. 

For $\varphi\in \{\sigma,\tau\}\subseteq \Omega$, $\pi\in \{\sigma,\tau\}\setminus \{\varphi\}$, $c\in [k]$ and $\{w_1,\dots, w_r\}=N(v)\cap \varphi^{-1}(c)$, we define the set 
$$
\cA_\varphi(c) := \{S_\varphi(v,c), \{S_\varphi(w_i,\pi(v))\}_{i\in [r]}\}\;.
$$
As we discussed before, it might be the case that $S_\varphi(w_i,\pi(v))=S_\varphi(w_j,\pi(v))$ for $i\neq j$. As $A_\varphi(c)$ is a set, we only consider this component once.
Define the multisets $\cA_\varphi:=\{\cA_\varphi(c):\,c\in [k]\}$ and $\overline{\cA_\varphi}:=\cK_\varphi\setminus \cA_\varphi$. These are multisets as we would like to count the Kempe component $S_\sigma(u,c)$ for each $c\in [k]$ that does not appear in the neighborhood of $u$. Since $\sigma$ and $\tau$ can be improper, it is possible that $\cA_\varphi(\pi(v))\cap \cA_\pi(\varphi(v))\neq \emptyset$. This is a case that must be treated separately and we refer to~\cite{V00} for it.

In general, the components in $\cA_\sigma$ and $\cA_\tau$ are different. However, since $\sigma$ and $\tau$ coincide in $V(G)\setminus \{v\}$, each $S\in \overline{\cA_\varphi}$ is not affected by $\varphi(v)$ and $\overline{\cA_\sigma}=\overline{\cA_\tau}$.

Define $a=1+a_1+\dots+a_r$ and $b=1+b_1+\dots+b_r$.
If $r\geq 1$, then let $a_{\max}=\max_{i\in [r]} a_i$ and $b_{\max}=\max_{i\in [r]} b_i$. 
Let $i_a$ denote an index $i\in [r]$ such that $a_{i}=a_{\max}$ and $i_b$ denote an index $i\in [r]$ such that $b_i=b_{\max}$. Define 
\begin{align}\label{eq:q}
q_i =\begin{cases}
p_{a_{\max}}-p_a & \text{if } i=i_a,\\
p_{a_i}& \text{otherwise},
\end{cases}
\hspace{3cm}
q_i' =\begin{cases}
p_{b_{\max}}-p_b & \text{if } i=i_b,\\
p_{b_i}& \text{otherwise}.
\end{cases}
\end{align}
Note that $q_i$ and $q_i'$ are non-negative since $p_{\ell+1}\leq p_{\ell}$ for $\ell\geq 1$.

Let $S(\sigma)$ be a random variable on $\cK_{\sigma} \cup\{\emptyset\}$ with probability distribution $\bP_\sigma$. 
In~\cite{V00}, Vigoda introduced a coupling  $(S(\sigma),S(\tau))$ on $(\cK_\sigma\cup\{\emptyset\})\times (\cK_\tau\cup\{\emptyset\})$ with probability distribution $\bP_{\sigma\tau}$ defined as follows.
\begin{itemize}
\item[i.]   if $S\in \cA_\sigma(c)$ for some $c\in [k]$, then
\begin{itemize}
\item[a)]  if $S=S_\sigma(v,c)$ and $r=0$, then $S=S_\tau(v,c)$ and  $\bP_{\sigma\tau}(S,S)=p_1/kn$.
\item[b)]  otherwise,
\begin{itemize}
\item[-] if $S= S_\sigma(v,c)$: let $S'=S_\tau(w_{i_a},\sigma(v))$, then $\bP_{\sigma\tau}(S,S')=p_a/kn$ and $\bP_{\sigma\tau}(S,S'')=0$ for each $S''\in \cK_\tau\setminus \{S'\}$.

\item[-] if $S= S_\sigma(w_{i},\tau(v))$ for $i=i_b$: let $S'=S_\tau(v,c)$, then $\bP_{\sigma\tau}(S,S')=p_b/kn$. Let $S''=S_\tau(w_i,\sigma(v))$, if $q'_i\leq q_i$, then $\bP_{\sigma\tau}(S,S'')=q_i'/kn$ and $\bP_{\sigma\tau}(\emptyset,S'')=(q_i-q_i')/kn$; otherwise,  $\bP_{\sigma\tau}(S,S'')=q_i/kn$ and $\bP_{\sigma\tau}(S,\emptyset)=(q_i'-q_i)/kn$.

\item[-] if $S= S_\sigma(w_{i},\tau(v))$ for $i\in [r]\setminus\{i_b\}$ with $b_i\neq 0$:  let $S''=S_\tau(w_i,\sigma(v))$, if $q'_i\leq q_i$, then $\bP_{\sigma\tau}(S,S'')=q_i'/kn$ and $\bP_{\sigma\tau}(\emptyset,S'')=(q_i-q_i')/kn$; otherwise,  $\bP_{\sigma\tau}(S,S'')=q_i/kn$  and $\bP_{\sigma\tau}(S,\emptyset)=(q_i'-q_i)/kn$.
\end{itemize}
\end{itemize}

\item[ii.] if $S\in \overline{\cA_\sigma}$: since $S\in \overline{\cA_\tau}$, then $\bP_{\sigma\tau}(S,S)=p_\ell/kn$, where $|S|=\ell$.

\end{itemize}

Given that $Y_t=\sigma$ and $Z_t=\tau$, the coupling $(S(\sigma),S(\tau))$ induces a coupling $(Y_{t+1},Z_{t+1})$ by setting $Y_{t+1}=\sigma_{S(\sigma)}$ and  $Z_{t+1}=\tau_{S(\tau)}$. This is the coupling we will use to prove Theorem~\ref{thm:improvement}.

\subsection{The key result}\label{sec:key}

Given the coupling introduced in the previous section, define the rescaled contributions to the expected change of $d_H$ and $d_B$ as
\begin{align*}
\nabla_H(\sigma, \tau) &:= kn\,\bE\left[d_H(Y_{t+1}, Z_{t+1})-d_H(\sigma, \tau) |Y_t=\sigma, Z_t=\tau\right] \\
&= kn\sum_{S\in\cK_\sigma\cup\{\emptyset\}\atop S'\in \cK_\tau\cup\{\emptyset\}} \bP_{\sigma \tau}(S,S') (d_H(\sigma_{S}, \tau_{S'})-d_H(\sigma, \tau)) \;,\\[0.3cm]
\nabla_B(\sigma, \tau) &:= -kn \,\bE\left[d_B(Y_{t+1}, Z_{t+1}) - d_B(\sigma, \tau) |Y_t=\sigma, Z_t=\tau\right]\\
&=- kn\sum_{S\in\cK_\sigma\cup\{\emptyset\}\atop S'\in \cK_\tau\cup\{\emptyset\}} \bP_{\sigma \tau}(S,S') (d_B(\sigma_S, \tau_{S'})-d_B(\sigma, \tau)) \;.
\end{align*}
The rescaling factor $kn$ is natural as the probability a Kempe component of size $\ell$ is flipped is exactly $p_\ell/kn$.

The total rescaled expected change can be written as
\begin{align}
\nabla(\sigma, \tau) &:= \nabla_H(\sigma, \tau)+\nabla_B(\sigma, \tau)= kn\bE\left[d(Y_{t+1}, Z_{t+1})-d(\sigma, \tau) |Y_t=\sigma, Z_t=\tau\right]\;. \label{eq:sum}
\end{align}

The crux of the argument to prove Theorem~\ref{thm:main} lies in showing that the expected change $\nabla(\sigma, \tau)$ is negative, as in the following theorem.
\begin{theorem}\label{thm:improvement}
There exists a bounded $\mathbf{p}$ such that if $k\geq \left(\frac{11}{6}-\eta\right)\Delta$, with $\eta=\frac{1}{\const}$, then for every $\sigma\tau\in E(\Gamma)$, the coupling defined in Section~\ref{sec:coup} satisfies 
$$
\nabla(\sigma,\tau)\leq -1\;.
$$
\end{theorem}
The choice of $-1$ in the theorem is arbitrary, and proving that $\nabla(\sigma,\tau)\leq c$ for any $c<0$ would be enough to show that the mixing time of flip dynamics is $O(n\log{n})$.

\section{Proofs of Theorems~\ref{thm:main_Glau} and~\ref{thm:main}}\label{sec:thm}
We now proceed with the proofs of our main results modulo Theorem~\ref{thm:improvement}, which we will prove in the next section.  We first analyze the mixing time of flip dynamics for $k$-colorings with flip parameters $\mathbf{p}$. The explicit values of the flip parameters will be set in Section~\ref{sec:nablaH}.
\begin{proof}[Proof of Theorem~\ref{thm:main}]
Let $\Omega_0 \subset \Omega = [k]^n$ be the set of proper $k$-colorings, and let $\sigma,\tau\in \Omega_0$ with $\sigma\tau\in E(\Gamma)$. Consider the metric $d$ on $\Omega$ obtained by extending the pre-metric $(\Gamma,\omega)$ defined in Section~\ref{sec:metric}. By Theorem~\ref{thm:improvement}, there exist a bounded $\mathbf{p}$ and $\eta>0$ such that if $k\geq \left(\frac{11}{6}-\eta\right)\Delta$, then flip dynamics for $k$-colorings with flip parameters $\mathbf{p}$ satisfies
 \begin{align*}
\bE\left[d(Y_{t+1}, Z_{t+1})|Y_t=\sigma, Z_t=\tau \right] \leq  d(\sigma, \tau) - \frac{1}{kn} \leq \left(1-\frac{1}{kn}\right) d(\sigma, \tau)\;.
\end{align*}
We can apply Theorem~\ref{thm:path coupling} with $\alpha=\frac{1}{kn}$ to conclude that the mixing time of the chain satisfies
$$
t_{\flip (\mathbf{p})}\leq kn\log{(4n)}\;.
$$
\end{proof}

Using the comparison theorem of Diaconis and Saloff-Coste~\cite{DSC}, Vigoda~\cite{V00} showed that for a particular choice of flip parameters $\mathbf{p}_{\textrm{Vig}}$ (see Lemma~\ref{lem:vig} for precise values) and $k\geq \frac{11}{6}\Delta$, the mixing time of flip dynamics for $k$-colorings with flip parameters $\mathbf{p}_{\textrm{Vig}}$ can be used to bound the mixing time of  Glauber dynamics for $k$-colorings. The proof relies on the fact that  $\mathbf{p}_{\textrm{Vig}}$ is bounded. 
It is straightforward to generalize Vigoda's result to other bounded $\mathbf{p}$ and smaller values of $k$.
\begin{theorem}\label{thm:trans}[Vigoda~\cite{V00}]
For every $\varepsilon>0$, bounded $\mathbf{p}$ and $k \geq (1+\varepsilon) \Delta$, the mixing time of Glauber dynamics for $k$-colorings  satisfies
$$
t_{\Glau} = O\left(n\log{k} \cdot t_{\flip(\mathbf{p})} \right)\;.
$$
\end{theorem}

Our main result easily follows from our result on flip dynamics and the previous theorem.
\begin{proof}[Proof of Main Theorem~\ref{thm:main_Glau}]
Let $\mathbf{p}$ be the flip parameters given by Theorem~\ref{thm:main}, thus $$t_{\flip (\mathbf{p})}\leq kn\log(4n).$$
Since $\mathbf{p}$ is bounded, we can apply Theorem~\ref{thm:trans} with $\varepsilon = \frac{5}{6}-\eta$, so $k\geq \left(\frac{11}{6}-\eta\right)\Delta = (1+\epsilon)\Delta$, to obtain that
$$
t_{\Glau} = O\left(\left(k\log{k}\right)\cdot n^2 \log{n} \right)\;.
$$
\end{proof}

\section{Proof of Theorem~\ref{thm:improvement}}\label{sec:change}

In this section, we prove Theorem~\ref{thm:improvement}. We will analyze the contributions from the Hamming metric $\nabla_H$ and the remainder $\nabla_B$ separately. 

\subsection{Contribution of $\nabla_H$ and choice of flip parameters}\label{sec:nablaH}

The core of Vigoda's argument in~\cite{V00} relies on bounding the expected change of the Hamming distance in one step of the coupling defined in Section~\ref{sec:coup}. In this section, we briefly describe Vigoda's analysis and refine his upper bound on $\nabla_H(\sigma,\tau)$ in terms of $\beta_{\sigma,\tau}(v)$.

\begin{lemma}[Lemma~5 in~\cite{V00}]\label{lem:vig}
Let $\textbf{p}_{\textrm{Vig}}=\left(1,\frac{13}{42},\frac{1}{6},\frac{2}{21},\frac{1}{21},\frac{1}{84},0 ,0 ,\dots\right)$.  
 Then, for every $\sigma \tau\in E(\Gamma)$, the coupling described in Section~\ref{sec:coup} satisfies
\begin{align*}
\nabla_H(\sigma, \tau) &\leq \frac{11}{6}\cdot \Delta-k\;.
\end{align*}
\end{lemma}
Recall that $p_0=0$. Vigoda's proof measures the contribution of each color to the expected change. 
For each $c\in [k]$, we define
\begin{align*}
\nabla_H(\sigma, \tau,c) &:= kn \sum_{S\in \cA_{\sigma}(c)\cup\{\emptyset\} \atop S'\in \cA_{\tau}(c)\cup\{\emptyset\}}\bP_{\sigma\tau}(S,S')(d_H(\sigma_S, \tau_{S'})-d_H(\sigma, \tau))\;.
\end{align*}
Consider also the remaining contribution,
\begin{align*}
\overline{\nabla_H}(\sigma, \tau) &:= kn \sum_{S\in \overline{\cA_{\varphi}}}\bP_{\sigma\tau}(S,S)(d_H(\sigma_S, \tau_S)-d_H(\sigma, \tau))\;.
\end{align*}
By the definition of the coupling, the Hamming distance does not change if $S\in \overline{\cA_\varphi}$, so $\overline{\nabla_H}(\sigma, \tau)=0$. Thus,
\begin{align}\label{eq:Hamm}
\nabla_H(\sigma, \tau) = \overline{\nabla_H}(\sigma, \tau) + \sum_{c\in [k]} \nabla_H(\sigma,\tau,c)=\sum_{c\in [k]} \nabla_H(\sigma,\tau,c)\;.
\end{align}
It suffices to bound $\nabla_H(\sigma,\tau,c)$. Lemma~\ref{lem:vig} follows directly from the choice of $p_\ell$ and the result from Vigoda's paper.

\begin{lemma}[see Eq.~(1) in~\cite{V00}]\label{lem:vig2}
 Suppose that $(\sigma,\tau)$ has an  $r$-configuration $(a_1,\dots,a_r;b_1,\dots,b_r)$ for $c$. If $r=0$, then  
\begin{align*}
\nabla_H(\sigma, \tau,c) =-1\;,
\end{align*}
and if $r\geq 1$, then
\begin{align}\label{eq:bound nabla_H}
\nabla_H(\sigma, \tau,c) \leq (a-a_{\max}-1)p_a+(b-b_{\max}-1)p_b+\sum_{i\in [r]} (a_i q_i+ b_i q_i'-\min\{q_i,q_i'\})\;.
\end{align}
\end{lemma}
Since $\sigma$ and $\tau$ might be improper, the previous lemma does not hold if $c\in\{\sigma(v),\tau(v)\}$. The main problem is that $S_\sigma(v,\tau(v))$ does not necessarily contain $S_\tau(w_i,\sigma(v))$ for every $i\in [r]$, and similarly for $S_\tau(v,\sigma(v))$. In this case, we refer to~\cite{V00} for the analysis of the coupling which yields $\nabla_H(\sigma, \tau,c)\leq r-1$, provided that $p_1\leq 1$.

In order to prove Lemma~\ref{lem:vig}, for a given $r$-configuration, we would like to obtain inequalities of the form $\nabla_H(\sigma,\tau,c)\leq r\kappa-1$, for a  constant $\kappa$ as small as possible.  Using the bound in Lemma~\ref{lem:vig2}, we obtain a collection of non-linear inequalities, as they involve $\min$ and $\max$ functions. One can easily set an LP problem $(P)$ such that any of its feasible solutions satisfies these equations, by adding an equation for each possible value of the minimum/maximum. Recall that $p_1=1$ and $p_{\ell+1}\leq p_\ell$ for $\ell\geq 1$.
For $i\geq 1$, let $r_i$ be the size of the configuration corresponding to the $i$-th inequality for some fixed enumeration of the constraints in Lemma~\ref{lem:vig2}. For $\ell \geq 1$, let $\alpha_{i \ell}$ be the coefficient of $p_\ell$ in it. We can describe $(P)$ as follows:
\begin{align*}
\begin{array}{lll@{}ll}
(P):&\text{minimize}  & \kappa &\\[0.2cm]
&\text{subject to}& \displaystyle\sum\limits_{\ell\geq 1}   &\alpha_{i \ell}p_{\ell} \leq r_i \kappa-1 ,  &i\geq 1\\
 &   &                                                &p_{\ell+1}-p_\ell\leq 0 , &\ell \geq  1,\\
 &   &                                                &p_{1}=1 , &\\
  &              &                                     &p_{\ell} \geq 0, & \ell \geq  2.
\end{array}\end{align*}
Consider the following reduced constraint linear program with constraints corresponding to all non-trivial $1$-configurations of size at most $6$ and to the $2$-configurations with $a_1=a_2 \in \{2,3\}$ and $b_1=b_2=1$.

\begin{align*}
\begin{array}{lll@{}ll}
(P_{\text{red}}):&\text{minimize}  & \kappa &\\[0.2cm]
&\text{subject to}&    &i(p_i - p_{i+1}) + (j-1)(p_j - p_{j+1}) \leq \kappa -1 , & 1 \leq i , j \leq 6,\, j\neq 1 \\
 &   &                                                &2(\ell-1)p_\ell  + p_{2\ell+1}+2 \leq 2\kappa -1 & \ell \in \{2,3\},\\
 &   &                                                &p_{\ell+1}-p_\ell\leq 0 , &\ell \geq  1,\\
 &   &                                                &p_{1}=1 , &\\
  &              &                                     &p_{\ell} \geq 0, & \ell \geq  2.
\end{array}\end{align*}
One can check that $\left(\textbf{p}_{\textrm{Vig}},\frac{11}{6}\right)$ forms an optimal solution to $(P_{\text{red}})$. In the discussion below we prove that a larger set of flip parameters, including $\textbf{p}_{\textrm{Vig}}$, correspond to feasible solutions of $(P)$. As $\left(\textbf{p}_{\textrm{Vig}},\frac{11}{6}\right)$  is optimal for $(P_{\text{red}})$ and $(P_{\text{red}})$ is a reduced version of $(P)$, it will be an optimal solution for $(P)$.  

The following technical statement is a compilation of the results from~\cite{V00} and follows from Lemma~\ref{lem:vig2}.
\begin{lemma}[\cite{V00}]~\label{lem:comp}
 Suppose that $(\sigma,\tau)$ has an $r$-configuration $(a_1,\dots,a_r;b_1,\dots,b_r)$ for $c$. If $ip_i \leq 1$ and $(i-1)p_i \leq 2p_3$, then  
\begin{itemize}
\item[(i)] if $r=0$, 
$\nabla_H(\sigma, \tau,c)=  -1;
$
\item[(ii)] if $r=1$, then 
$
\nabla_H(\sigma, \tau,c) \leq  
a_1 (p_{a_1}-p_{a_1+1})+b_1(p_{b_1}-p_{b_1+1})- \min(p_{a_1}-p_{a_1+1},p_{b_1}-p_{b_1+1});
$
\item[(iii)] if $r=2$, then 
$
\nabla_H(\sigma, \tau,c)\leq  2(\ell-1)p_{\ell}+p_{2\ell+1}+2,
$
for $\ell \in \{2,3\}$, moreover the equality only holds if the configuration is either $(\ell,\ell;1,1)$ or $(1,1;\ell,\ell)$;
\item[(iv)] if $r\geq 3$, then 
$
\nabla_H(\sigma, \tau,c)\leq (a-2\,  a_{\max})p_a+(b-2\,  b_{\max})p_b +r(p_1+2p_3). 
$
\end{itemize}
\end{lemma}

If $(\mathbf{p},\kappa)$ is a feasible solution of $(P_{\text{red}})$, then, by Lemma~\ref{lem:comp}, the constraints in $(P)$ corresponding to $r$-configurations are satisfied when $r\in\{0,1,2\}$. Next result shows that, under some technical conditions on the flip parameters, the solution $(\mathbf{p},\kappa)$ also satisfies the constraints for $r\geq 3$.

\begin{lemma}\label{lem:trips}
If $ip_i \leq 1$, $(i-1)p_i \leq 2p_3$, $p_1 + 2p_3 = \frac{4}{3}<\kappa$, and $(i-2)p_i < \frac{1}{4} - \frac{3}{2} \left(\frac{11}{6} - \kappa\right)$, 
then for all $r\ge 3$
$$\nabla_H(\sigma, \tau,c) < r\kappa -1\;.$$
\end{lemma}
\begin{proof}
Because $a_{\max},b_{\max}\geq 1$ and $p_1 + 2p_3 = \frac{4}{3}$, we see that for $r=3$
\begin{eqnarray*}
\nabla_H(\sigma, \tau,c) &\leq& (a-2\,  a_{\max})p_a+(b-2\,  b_{\max})p_b +3(p_1+2p_3)\\ 
&<& \frac{1}{2} - 3\left(\frac{11}{6}-\kappa\right) +3(p_1+2p_3)\\
&=& \frac{9}{2} - 3\left(\frac{11}{6}-\kappa\right)= 3\kappa -1\;.
\end{eqnarray*}

Because $p_1+2p_3 = \frac{4}{3} < \kappa$, for all $r \geq 3$,
$$\nabla_H(\sigma, \tau,c) \leq (a-2\,  a_{\max})p_a+(b-2\,  b_{\max})p_b +r(p_1+2p_3) < r\kappa -1\;.$$
\end{proof}

By Lemma~\ref{lem:comp}, we conclude that any feasible solution of $(P_{\text{red}})$ that satisfies the conditions of Lemma~\ref{lem:trips} is a feasible solution of $(P)$; in particular $\left(\textbf{p}_{\text{Vig}},\frac{11}{6}\right)$ is an optimal solution of $(P)$.

Given a solution $(\mathbf{p},\kappa )$ of $(P)$, we say that an $r$-configuration is \emph{$\textbf{p}$-extremal} if $\nabla_H(\sigma,\tau, c)=\kappa r-1$ for flip dynamic with flip parameters $\mathbf{p}$. Since $\mathbf{p}_{\text{Vig}}$ satisfies the hypothesis of Lemma~\ref{lem:trips}, there are no $\mathbf{p}_{\text{Vig}}$-extremal $r$-configurations, for $r\geq 3$. By Lemma~\ref{lem:comp}~$(iii)$, any $\mathbf{p}_{\text{Vig}}$-extremal $2$-configuration is of the form $(\ell,\ell;1,1)$ or $(1,1;\ell,\ell)$, for some $\ell \in \{2,3\}$. A simple computation shows that, up to symmetries, there are six $\mathbf{p}_{\text{Vig}}$-extremal configurations: $(2;1)$, $(3;1)$, $(4;1)$, $(5;1)$, $(2,2;1,1)$, and $(3,3;1,1)$.

In order to simplify the analysis, we would like to find an optimal solution $\left(\mathbf{p},\frac{11}{6}\right)$ of $(P)$ that minimizes the number of $\mathbf{p}$-extremal configurations. There are two crucial constraints in $(P)$ which correspond to the extremal configurations defined in Section~\ref{sec:metric}. The extremal $1$-configurations $(2;1)$ and $(1;2)$ lead to the inequality $p_1-p_3\leq \kappa-1$ and the extremal $2$-configurations $(3,3;1,1)$ and $(1,1;3,3)$ lead to $4p_3+p_7\leq 2\kappa-3$. As $p_1=1$ and $p_7\geq 0$, these two inequalities already imply that $\kappa \ge \frac{11}{6}$. Moreover, if $\kappa = \frac{11}{6}$, then $p_3=\frac{1}{6}$ and $p_7=0$.
As we will show, there exist optimal solutions $\left(\mathbf{p},\frac{11}{6}\right)$ of $(P)$ with only two $\mathbf{p}$-extremal configurations, up to symmetries, corresponding to $(2;1)$ and $(3,3;1,1)$.

This motivates the introduction of the LP problems $(P^*)$ and $(P^*_{\text{red}})$ with the same variables, optimization function, and constraints as $(P)$ and $(P_{\text{red}})$, respectively, apart from the constraints given by the extremal $1$-configurations $(2;1)$ and $(1;2)$ and by the extremal $2$-configurations $(3,3;1,1)$ and $(1,1;3,3)$, which are replaced by $p_3=\frac{1}{6}$ and $p_7=0$. Hence, $p_{\ell}=0$ for all $\ell\geq 7$ and the set of variables is now finite. Clearly, if $(\mathbf{p^*},\kappa^*)$ is an optimal solution for $(P^*)$, then $\left(\mathbf{p^*},\frac{11}{6}\right)$ is an optimal solution for $(P)$, and similarly for the reduced version.

\begin{observation}\label{obs}
The program $(P^*_{\textrm{red}})$ has optimal solution $\mathbf{p^*}=(p^*_1,\dots,p^*_6,\kappa^*)$ with
$$
p^*_1 = 1,\, p^*_2 = \frac{185}{616},\, p^*_3=\frac{1}{6},\, p^*_4=\frac{47}{462},\, p^*_5=\frac{9}{154},\, p^*_6 =\frac{2}{77} \text{ and }\kappa^*=\frac{161}{88}\;.
$$
\end{observation}
The values of $\mathbf{p^*}$ given in this observation are the values of the flip parameters we will use in the proof of Theorem~\ref{thm:improvement}.  Note that these values satisfy $ip_i \leq 1$, $(i-1)p_i \leq 2p_3$, $p_1 + 2p_3 = \frac{4}{3}<\kappa$ and $(i-2)p_i <  \frac{1}{4} - \frac{3}{2}\left(\frac{11}{6} - \kappa^*\right)=\frac{1}{4} - \frac{3}{2}\left(\frac{11}{6} - \frac{161}{88}\right) = \frac{43}{176}$.  Hence, by Lemma~\ref{lem:trips} with $\kappa = \kappa^*$, we find that, for all $r \ge 3$, $\nabla_H(\sigma, \tau,c) <  r\kappa^* -1$. One can verify using a computer that $\nabla_H(\sigma, \tau,c) \leq r\kappa^* -1$ for all $r$-configurations with $r\in\{0,1,2\}$ other than $(2;1)$, $(1;2)$, $(3,3;1,1)$, and $(1,1;3,3)$.  So $(p^*,\kappa^*)$ is an optimal solution of $(P^*)$ and of $(P)$, and  up to symmetries, there are only two  $\mathbf{p^*}$-extremal configurations, $(2;1)$ and $(3,3;1,1)$.

For flip dynamics with flip parameters $\mathbf{p^*}$, it follows that 
if $c\in B_{\sigma,\tau}(v)$, then
$$
\nabla_H(\sigma, \tau,c) \leq \frac{11}{6}\;,
$$
and that if $c\notin B_{\sigma,\tau}(v)$, then
$$
\nabla_H(\sigma, \tau,c) \leq \kappa^*=\frac{161}{88}\;.
$$

Our next lemma follows directly from these two equations and~\eqref{eq:Hamm}.
\begin{lemma}\label{lem:improvement}
Let $\varepsilon=\frac{11}{6}-\frac{161}{88}$. For every $\sigma \tau\in E(\Gamma)$, we have
$$
\nabla_H(\sigma, \tau) \leq \left(\frac{11}{6}-\varepsilon\left(1-\beta_{\sigma,\tau}\left(v\right)\right)\right) \Delta-k\;.
$$
\end{lemma}

\subsection{Contribution of $\nabla_B$}\label{sec:nablaB}

In this section we bound $\nabla_B(\sigma,\tau)$ from above. Recall the coupling defined in Section~\ref{sec:coup}. Similarly as before, we define the contributions
\begin{align*}
\nabla_B(\sigma, \tau, c) &:=- kn\sum_{S\in\cA_\sigma(c)\cup\{\emptyset\}\atop S'\in \cA_\tau(c)\cup\{\emptyset\}} \bP_{\sigma \tau}(S,S') (d_B(\sigma_S, \tau_{S'})-d_B(\sigma, \tau)) \;,\\
\overline{\nabla_B}(\sigma, \tau) &:=- kn\sum_{S\in\overline{\cA_\varphi}} \bP_{\sigma \tau}(S,S) (d_B(\sigma_S, \tau_{S})-d_B(\sigma, \tau)) \;.
\end{align*}
By equations (\ref{eq:omega}) and (\ref{eq:d_B}), since $d_H(\sigma,\tau)=1$, we have $d_B(\sigma,\tau)= 1-\omega(\sigma,\tau)=\gamma (1-\beta_{\sigma,\tau}(v))$. Moreover, $d_B(\sigma_S,\tau_{S'})\geq 0$, for every $S\in \cK_\sigma$ and $S'\in \cK_\tau$. By the properties of the coupling,
\begin{align*}
\nabla_B(\sigma, \tau, c) &\leq -\gamma(1-\beta_{\sigma,\tau}(v)) kn  \sum_{S\in\cA_\sigma(c)\cup\{\emptyset\}\atop S'\in \cA_\tau(c)\cup\{\emptyset\}} \bP_{\sigma \tau}(S,S')  \\
 &\leq -\gamma(1-\beta_{\sigma,\tau}(v)) kn \left(\sum_{S\in\cA_\sigma(c)} \bP_{\sigma}(S) +\sum_{S'\in\cA_\tau(c)} \bP_{\sigma}(S')\right) \\
  &= -2\gamma(1-\beta_{\sigma,\tau}(v)) (|N(v)\cap \varphi^{-1}(c)|+1) \;,
\end{align*}
where we have used that $\bP_\sigma (S)\leq 1/kn$.

We can bound the expected change of $\nabla_B$ as follows
\begin{align}\label{eq:bound nabla_B}
\nabla_B(\sigma, \tau) &= \overline{\nabla_B}(\sigma, \tau) + \sum_{c\in [k]} \nabla_B(\sigma,\tau,c)\leq \overline{\nabla_B}(\sigma, \tau)+2\gamma (k+\Delta) (1-\beta_{\sigma,\tau}(v))  \;.
\end{align}

An important difference here as opposed to the analysis of the contribution of $\nabla_H$, is that the components in $\overline{\cA_\varphi}$ have an effect on the expected change of $\nabla_B$. For $S\in \overline{\cA_\varphi}$, since $d_H(\sigma_S,\tau_S)=1$, we have $d_B(\sigma_S,\varphi_S)= \gamma(1-\beta_{\sigma_S,\tau_S}(v))$.
It follows that,
\begin{align*}
\overline{\nabla_B}(\sigma, \tau) &=\gamma \sum_{S\in\overline{\cA_\varphi}} p_{|S|} (\beta_{\sigma_S, \tau_S}(v)-\beta_{\sigma, \tau}(v)) \;.
\end{align*}
For each $c\in [k]$ and $i\in\{1,2\}$ and $S\in\overline{\cA_\varphi}$, let 
$$
\xi_{\sigma,\tau}(v,c,S) := 
\begin{cases}
-i & \text{if }c\in B^i_{\sigma,\tau}(v) \text{ and } c\notin B_{\sigma_S,\tau_S}(v),\\
i & \text{if }c\notin B_{\sigma,\tau}(v)  \text{ and }c\in B^i_{\sigma_S,\tau_S}(v),\\
-1 & \text{if }c\in B^2_{\sigma,\tau}(v) \text{ and } c\in B^1_{\sigma_S,\tau_S}(v),\\
1 & \text{if }c\in B^1_{\sigma,\tau}(v)  \text{ and }c\in B^2_{\sigma_S,\tau_S}(v), \\
0 & \text{otherwise}.
\end{cases}
$$
The variable $\xi_{\sigma,\tau}(v,c,S)$ can be understood as the contribution of color $c$ to $\beta_{\sigma_S, \tau_S}(v)-\beta_{\sigma, \tau}(v)$.
For every $\cS\subseteq \overline{\cA_\varphi}$, we define
$$
\overline{\nabla_B}(\sigma, \tau, c, \cS) = \frac{\gamma}{\Delta}  \sum_{S\in \cS} p_{|S|} \xi_{\sigma,\tau}(v,c,S)\;,
$$
and note that 
$$
\overline{\nabla_B}(\sigma, \tau) = \sum_{c\in [k]} \overline{\nabla_B}(\sigma, \tau, c, \overline{\cA_\varphi})  \;.
$$

Next lemma bounds from above the contribution of each $\overline{\nabla_B}(\sigma, \tau, c, \overline{\cA_\varphi})$.
\begin{lemma}\label{lem:complement}
For $i\in\{1,2\}$, if $c\in B^i_{\sigma,\tau}(v)$, then
$$
\overline{\nabla_B}(\sigma, \tau, c,\overline{\cA_\varphi}) \leq -i \gamma\left( \frac{k}{\Delta}-\frac{3}{2}\right)\;.
$$
If $c\notin B_{\sigma,\tau}(v)$, then
$$
\overline{\nabla_B}(\sigma, \tau, c,\overline{\cA_\varphi}) \leq  2\gamma\left(9+\frac{15k}{\Delta}\right)\;.
$$
\end{lemma}
\begin{proof}

Recall that the extremal configurations are $(2;1)$, $(1;2)$, $(3,3;1,1)$ and $(1,1;3,3)$ and that $W=N(v)\cap \varphi^{-1}(c)=\{w_1,\dots,w_r\}$.

Assume first that $c\in B^i_{\sigma,\tau}(v)$ for some $i\in \{1,2\}$. Consider the sets of components
\begin{align*}
\cS_0&:= \{S\in \overline{\cA_\varphi}: c\notin B_{\sigma_S,\tau_S}(v)\}\;,\\
\cS_2&:= \{S\in \overline{\cA_\varphi}: c\in B^2_{\sigma_S,\tau_S}(v)\}\;.
\end{align*}
Note that when $i=1$, then for every $S\in \overline{\cA_\varphi}\setminus  (\cS_0\cup \cS_2)$ we have $\xi_{\sigma,\tau}(v,c,S)\leq 0$; therefore,
$$
\overline{\nabla_B}(\sigma, \tau, c,\overline{\cA_\varphi}) \leq 
\overline{\nabla_B}(\sigma, \tau, c,\cS_0)+ \overline{\nabla_B}(\sigma, \tau, c,\cS_2)\;.
$$
Note that when $i=2$, then for every $S\in \overline{\cA_\varphi} \setminus \cS_0$ we have $\xi_{\sigma,\tau}(v,c,S)\leq 0$; therefore,
$$
\overline{\nabla_B}(\sigma, \tau, c,\overline{\cA_\varphi}) \leq 
\overline{\nabla_B}(\sigma, \tau, c,\cS_0) \;.
$$

We proceed to bound $\overline{\nabla_B}(\sigma, \tau, c,\cS_0)$ for $i\in\{1,2\}$.
Without loss of generality, assume that $a_1>b_1$. Let $u\in S_\tau(w_1,\sigma(v))$ with $\tau(u)=\sigma(v)$; we note that  $u\notin W\cup\{v\}$ and that such a vertex always exists as $a_1\geq 2$.
Choose a color $c'\in [k]$ with $c'\notin \varphi(N(u))\cup\{\sigma(v),\tau(v)\}$. Let $S=S_\varphi(u,c') \in \overline{A_\varphi}$. As $S=\{u\}$, $(\sigma_{S},\tau_S) $ has either a $(1;1)$ or a $(j,3;1,1)$ (with $j\in\{1,2\}$) configuration for $c$, i.e.  $c\notin B_{\sigma_S,\tau_S}(v)$. As there are at least $k-\Delta-2$ choices for $c'$ and as $p_{|S|}=p_1=1$, we have
$$
\overline{\nabla_B}(\sigma, \tau, c,\cS_0)\leq - \frac{\gamma(k-\Delta-2)}{\Delta
}\cdot i\;.
$$

Now we bound $\overline{\nabla_B}(\sigma, \tau, c,\cS_2)$, provided that $i=1$. Let $S\in \cS_2$, then $|S\cap (N(v)\setminus\{w_1\})|\geq 1$ and if $w\in S\cap (N(v)\setminus \{w_1\})$, then $\varphi_S(w)=c$. Thus, $S$  can be described as $S=S_\varphi(w,c)$ for $w\in N(v)$, implying that $|\cS_2|\leq \Delta$.  Moreover, $|S|\geq 2$ as at least two vertices need to change their color to transform an extremal $1$-configuration into an extremal $2$-configuration. Since $p_{|S|}\leq p_2\leq \frac{1}{3}$ and $\xi_{\sigma,\tau}(v,c,S)=1$, we have
$$
\overline{\nabla_B}(\sigma, \tau, c,\cS_2)\leq \frac{\gamma}{3}\;.
$$

From the bounds on $\overline{\nabla_B}(\sigma, \tau, c,\cS_0)$ and $\overline{\nabla_B}(\sigma, \tau, c,\cS_2)$ derived above, we obtain that for $i\in\{1,2\}$ and $c\in B^i_{\sigma,\tau}(v)$
$$
\overline{\nabla_B}(\sigma, \tau, c,\overline{\cA_\varphi})\leq- \frac{\gamma \left(k- \frac{4\Delta}{3}-2\right)}{\Delta} \cdot i \leq -i \gamma\left( \frac{k}{\Delta}-\frac{3}{2}\right) \;,
$$
and this proves the first statement.

\medskip

To prove the second statement, assume that $c\notin B_{\sigma,\tau}(v)$ and let $\cT:=\{S\in \overline{\cA_\varphi}: c\in B_{\sigma_S,\tau_S}(v)\}$. 
Again, for every $S\in \overline{\cA_\varphi} \setminus \cT$, we have $\xi_{\sigma,\tau}(v,c,S)\leq 0$. Therefore,
$$
\overline{\nabla_B}(\sigma, \tau, c,\overline{\cA_\varphi}) \leq 
\overline{\nabla_B}(\sigma, \tau, c,\cT)\;.
$$
Define $W_S=N(v)\cap \varphi_S^{-1}(c)$ with $|W_S|=r_S$ and note that $r_S\leq 2$.
Consider the partition $\cT=\cT_1\cup\cT_2\cup \cT_3$ with
\begin{align*}
\cT_1&:= \{S\in \overline{\cA_\varphi}: W\setminus W_S\neq \emptyset\}\;,\\
\cT_2&:= \{S\in \overline{\cA_\varphi}: W_S\setminus W\neq \emptyset\}\setminus \cT_1\;,\\
\cT_3&:= \{S\in \overline{\cA_\varphi}: W_S= W\}\;.
\end{align*}
For every $S\in \cT_3$, if $c\in B^1_{\sigma_S,\tau_S}(v)$, let $((a_S)_1;(b_S)_1)$ be the extremal $1$-configuration for $c$ in $(\sigma_S,\tau_S)$ and if $c\in B^2_{\sigma_S,\tau_S}(v)$, let $((a_S)_1,(a_S)_{2};(b_S)_1,(b_S)_{2})$ be the extremal $2$-configuration for $c$ in $(\sigma_S,\tau_S)$.
Recall that $(a_1,\dots,a_r;b_1,\dots,b_r)$ is the $r$-configuration for $c$ in $(\sigma,\tau)$. As it is non-extremal, there exists $x\in\{a,b\}$ and $j\in [r_S]$, such that $x_j\neq (x_S)_j$. Note that $(x_S)_j\leq 3$. 

Consider the partition $\cT_3=\cT_3^+ \cup\cT_3^- $ with
\begin{align*}
\cT_3^+&:= \{S\in \cT_3: x_j> (x_S)_j\}\;,\\
\cT_3^-&:= \{S\in \cT_3: x_j< (x_S)_j\}\;.
\end{align*}
To bound the size of $\cS\in \{\cT_1,\cT_3^+\}$ we will proceed as follows. For every $S\in \cS$, there is a vertex in a Kempe component of either $\sigma$ or $\tau$ that does not belong to the corresponding component in either $\sigma_S$ or $\tau_S$.  If there exists $R(\cS)\subseteq S_\sigma(v,c)\cup S_\tau(v,c)$ such that $S\cap R(\cS)\neq \emptyset$ for every $S\in \cS$, then, any $S\in \cS$ can be described as $S=S_\varphi(u,c')$ for $u\in R(\cS)$ and $c'\in [k]$, and $|\cS|\leq |R(\cS)| k$.

If $\cS= \cT_1$ and $S\in \cS$, then observe that $|S\cap W|= |W\setminus W_S|\geq \max\{r-r_S,1\}$. Let $m=\min\{r_S+1, r\} $. If $R(\cT_1)=R_1=\{w_1,\dots, w_{m}\}$, it follows that $|S\cap R_1|\geq |S\cap W|- (r-(r_S+1)) \geq 1$ and $|\cT_1|\leq  (r_S+1) k\leq 3 k$. 

If $\cS= \cT_3^+$ and $S\in \cS$, recall that $x_j>(x_S)_j$ and set $\varphi=\sigma$ if $x=b$ and  $\varphi=\tau$ if $x=a$, and let $\pi\in\{\sigma,\tau\}\setminus  \{\varphi\}$.
Let $R(\cT_3^+)=R_3$ be an arbitrary set of $(x_S)_j$ vertices in  $S_{\varphi}(w_j,\pi(v)) \setminus \left\{w_j\right\}$. As $w_j\notin R_3$, we have $S\cap R_3\neq \emptyset$.  
Since there are $4$ choices for the extremal configuration, we have $|\cT_3^+|\leq 4 (x_S)_j k\leq  12 k$.
\medskip

To bound the size of $\cS\in \{\cT_2,\cT_3^-\}$ we will proceed as follows. For every $S\in \cS$, there is a vertex in the neighborhood of a Kempe component of either $\sigma$ or $\tau$, that belongs to the corresponding component in either $\sigma_S$ or $\tau_S$.  If there exists a set $N(\cS)$ of neighbors of $S_\varphi(v,c)$ such that $S\cap N(\cS)\neq \emptyset$ for every $S\in \cS$, then, any $S\in \cS$ can be described as $S=S_\varphi(u,c')$ for $u\in N(\cS)$ and a unique $c'\in\{c,\pi(v)\}$, and $|\cS|\leq |N(\cS)|$.

If $\cS= \cT_2$ and $S\in \cS$, then let $N(\cT_2)=N_2=N(v)\setminus W$. Clearly $S\cap N_2 \neq \emptyset$ and 
$|\cT_2|\leq \Delta$.

If $\cS= \cT_3^-$ and $S\in \cS$,  recall that $x_j<(x_S)_j$ and set $\varphi=\sigma$ if $x=b$ and  $\varphi=\tau$ if $x=a$, and let $\pi\in\{\sigma,\tau\}\setminus  \{\varphi\}$.
Let $N(\cT_3^-)=N_3$ be the set of neighbors of $S_\varphi (w_j,\pi(v)) $, which satisfies $S\cap N_3\neq\emptyset$. 
As $S\in \cT_3^-$, $|S|\leq x_j\Delta\leq ((x_S)_j-1)\Delta \leq 2\Delta$.  Since there are $4$ choices for the extremal configuration, we have $|\cT_3^-|\leq 8\Delta$.
\medskip

Since $p_{|S|}\leq 1$ and $\xi_{\sigma,\tau}(v,c,S)\leq 2$, we conclude the second statement of the lemma,
$$
\overline{\nabla_B}(\sigma, \tau, c,\overline{\cA_\varphi})\leq \overline{\nabla_B}(\sigma, \tau, c,\cT) \leq \frac{2\gamma}{\Delta}\left(3k+12k+\Delta+8 \Delta\right) = 2\gamma\left(9+\frac{15k}{\Delta}\right)\;.
$$
\end{proof}

The following bound on $\nabla_B$ follows directly from~\eqref{eq:bound nabla_B} and Lemma~\ref{lem:complement}.
\begin{lemma}\label{lem:improvement2}
For every $\sigma\tau\in E(\Gamma)$, we have
$$ 
\nabla_B(\sigma, \tau) \leq - \gamma\left(\frac{k}{\Delta}-\frac{3}{2}\right)\beta_{\sigma,\tau}(v)+2\gamma\left(10+\frac{16k}{\Delta}\right)(1-\beta_{\sigma,\tau}(v))\;.
$$
\end{lemma}

We conclude this section with the proof of Theorem~\ref{thm:improvement}.
\begin{proof}[Proof of Theorem~\ref{thm:improvement}]
Recall that $\varepsilon =\frac{11}{6}-\frac{161}{88} = \frac{1}{264}$, and set $\gamma = \frac{\varepsilon\Delta}{53 k}$ and 
$$
k\geq  \left(\frac{11}{6}-\frac{1}{84000}\right)\Delta \geq \left(\frac{11}{6}-\frac{\varepsilon}{318}\right)\Delta +1\;,
$$ 
provided that $\Delta$ is large enough. Note that $\frac{k}{\Delta}\geq \frac{9}{5}$.
Using~\eqref{eq:sum}, Lemma~\ref{lem:improvement} and~\ref{lem:improvement2}, it follows that  
\begin{align*}
\nabla(\sigma,\tau)&\leq \left(\frac{11}{6} - \left(\varepsilon-2\gamma\left(10+\frac{16k}{\Delta}\right)\right)(1 - \beta_{\sigma,\tau}(v)) - \gamma \left(\frac{k}{\Delta}-\frac{3}{2}\right) \beta_{\sigma,\tau}(v)\right)\Delta  -k \\
&\leq \left(\frac{11}{6} - \left(\varepsilon-\frac{52\gamma k}{\Delta}\right)(1 - \beta_{\sigma,\tau}(v)) - \frac{\gamma k}{6 \Delta}  \beta_{\sigma,\tau}(v)\right)\Delta  -k \\
&= \left(\frac{11}{6} - \frac{\varepsilon}{53}(1 - \beta_{\sigma,\tau}(v)) - \frac{\varepsilon}{318} \beta_{\sigma,\tau}(v)\right)\Delta - k\\
&\leq \left(\frac{11}{6} - \frac{\varepsilon}{318} \right)\Delta  - k \leq -1\;,
\end{align*}
as desired.
\end{proof}

\section{List coloring}\label{sec:list}

In this section we show rapid mixing for the list coloring version of Glauber dynamics for the same range of $k$ as in the non-list colorings, thus giving a proof of Theorem~\ref{thm:list_Glau}.

A \emph{list assignment} of $G$ is a function $L: V(G) \rightarrow 2^\mathbb{N}$. An \emph{$L$-coloring} is a function $\sigma:V(G)\rightarrow \mathbb{N}$ such that $\sigma(u) \in L(u)$ for all $u \in V(G)$.  Usually in the literature list colorings are assumed to be proper, here we will not require this but distinguish between proper and not necessarily proper list colorings.  We denote by $\Omega^{L}$ the set of all $L$-colorings of $G$.
If $|L(u)| = k$ for all $u \in V(G)$, then we say that $L$ is a \emph{$k$-list-assignment} and that an $L$-coloring is a \emph{$k$-list-coloring}. 

The \emph{Glauber dynamics for $L$-colorings} is a discrete-time Markov chain $(X^L_t)$ with state space $\Omega^{L}$ and transitions between states given by recoloring at most one vertex; if $X^L_t=\sigma$, then we proceed as follows.
\begin{enumerate}
\setlength\itemsep{0em}
\item Choose $u$ uniformly at random from $V(G)$.
\item For all vertices $v \neq u$, let $X^L_{t+1}(v) = \sigma(v)$.
\item Choose $c$ uniformly at random from $L(u)$, if $c$ does not appear among the colors in the neighborhood of $u$ then let $X^L_{t+1}(u) = c$, otherwise let $X^L_{t+1}(u) = \sigma(u)$.
\end{enumerate}
Although the state space is $\Omega^L$, $(X_t^L)$ will converge to the uniform distribution on proper $L$-colorings (we refer to the discussion at the end of Section~\ref{sec:flip} for further details).

The proof strategy to show that Glauber dynamics for $k$-list-colorings is rapidly mixing provided that $k$ is large enough will be analogous to the non-list coloring case. 

Before describing the version of flip dynamics for list colorings that we will analyze, we introduce some definitions.  Given $\sigma\in \Omega^{L}$, one can define Kempe components of $\sigma$ as for colorings and we denote by $\cK^L_\sigma$ the multiset of Kempe components $S=S_\sigma(u,c)$ with $u\in V(G)$ and $c\in L(u)$. As before, the Kempe components $S=\{u\}$ are counted with multiplicity for each color $c\in L(u)$ that does not appear in the neighborhood of $u$.

Recall that $\sigma_S$ is obtained by swapping the colors in $S$ and note that $\sigma_S$ is not necessarily an $L$-coloring as the new color of a vertex might not be in its list. Given a  Kempe component $(c_1,c_2,S)$ in $\cK^L_\sigma$, we say that $S$ is \emph{flippable} if for every $u\in S$ we have $\{c_1,c_2\}\subseteq L(u)$. If $S$ is flippable, then $\sigma_S\in \Omega^L$.

The \emph{flip dynamics for $L$-colorings with flip parameters $\mathbf{p}=(p_1,p_2,\dots)$} is a discrete-time Markov chain $(Y^L_t)$ with state space $\Omega^L$ and transitions between states given by swapping colors in flippable Kempe components; if $Y^L_t=\sigma$, then we proceed as follows.
\begin{enumerate}
\setlength\itemsep{0em}
\item Choose $u$ uniformly at random from $V(G)$.
\item Choose $c$ uniformly at random from $L(u)$.
\item Let $S=S_{\sigma}(u,c)$ and $\ell= |S|$. If $S$ is flippable, with probability $p_\ell/\ell$ let $Y^L_{t+1}= \sigma_{S}$, otherwise 
let $Y^L_{t+1}=\sigma$.
\end{enumerate}

We will prove the analogous version of Theorem~\ref{thm:main} for list colorings.
\begin{theorem}\label{thm:list_flip}
There exists a bounded $\mathbf{p}$ such that for every $k\geq \left(\frac{11}{6}-\eta\right)\Delta$, with $\eta = \frac{1}{\const}$, and every $k$-list-assignment $L$, flip dynamics for $L$-colorings on a graph on $n$ vertices with maximum degree $\Delta$ has mixing time 
$$
t_{L-\flip (\mathbf{p})}\leq kn\log{(4 n)}\;.
$$
\end{theorem}

The proof of this theorem follows the same lines as Theorem~\ref{thm:main}. We will describe the proof strategy, stressing the parts where the argument is different for list coloring and omitting the ones that are straightforward adaptations of the  coloring case.

Let $\sigma,\tau\in \Omega^L$ that differ only at a vertex $v$. 
For $\varphi\in \{\sigma,\tau\}\subseteq \Omega^L$, $\pi\in \{\sigma,\tau\}\setminus \{\varphi\}$, $c\in L(v)$ and $\{w_1,\dots, w_r\}=N(v)\cap \varphi^{-1}(c)$, we define the $r$-configurations $(a^L_1,\dots,a^L_r;b^L_1,\dots,b^L_r)$ for $c$ in $(\sigma,\tau)$ as before, with the sole difference that we also set $a^L_i=0$ if $S_\tau(w_i,\sigma(v))$ is not flippable and $b^L_i=0$ if $S_\sigma(w_i,\tau(v))$ is not flippable.  We define $i^L_a$, $i^L_b$, $a^L_{\max}$ and $b^L_{\max}$ analogously as before, and note that the latter two can be zero.
Let $a^L=1+ a_1^L+\dots+a_r^L$ if $S_\sigma(v,c)$ is flippable and $a^L=0$ otherwise. Let $b^L=1+b_1^L+\dots+b^L_r$ if $S_\tau(v,c)$ is flippable and $b^L=0$ otherwise. Define $q_i(L)$ and $q_i'(L)$ as in~\eqref{eq:q} for the list version of the parameters.

According to this, we use the same definition of extremal configurations, metric $d$ on $\Omega^L$,  $d_H$ and $d_B$. Again, for any pair $\sigma',\tau'\in \Omega^L$, we have $d(\sigma',\tau') \leq d_H(\sigma',\tau')$, which implies that $d_B(\sigma',\tau')\geq 0$.

For $c\in \mathbb{N}$ consider the sets
$\cA^L_\varphi(c) := \{S_\varphi(v,c), \{S_\varphi(w_i,\pi(v))\}_{i\in [r]}\}$ and the multisets $\cA^L_\varphi:=\{\cA^L_\varphi(c):\,c\in \mathbb{N}\}$ and $\overline{\cA^L_\varphi}:=\cK_\varphi^L\setminus \cA^L_\varphi$. 
As before, for every $S\in \cK^L_\sigma\cup \{\emptyset\}$, one can define $\bP^L_\sigma(S):=\bP(Y^L_{t+1}=\sigma_S\mid \, Y^L_t=\sigma)$. 
We use the same coupling as the one defined in Section~\ref{sec:coup} and define $\nabla^L$, $\nabla_H^L$ and $\nabla_B^L$ analogously as for colorings. Fix the flip parameters $\mathbf{p^*}$ provided in Observation~\ref{obs}.

We will prove an analogue of Lemma~\ref{lem:improvement} to bound $\nabla^L_H$ for list colorings. As in Section~\ref{sec:nablaH}, we have
$$
\nabla^L_H(\sigma,\tau) = \sum_{c\in \mathbb{N}} \nabla_H^L(\sigma,\tau,c)\;.
$$
Suppose first that $r=0$. Then $c\in L(v)$ and $\nabla_H^L(\sigma,\tau,c)=-1$. 
If $r\geq 1$, the analogous of equation~\eqref{eq:bound nabla_H} also holds for list colorings,
\begin{align}\label{eq:list}
\nabla_H^L(\sigma,\tau,c) &\leq (a^L-a^L_{\max}-1)p_{a^L}+(b^L-b^L_{\max}-1)p_{b^L}\nonumber\\ 
&\;\;\;\;+ \sum_{i\in [r]} (a^L_i q_i(L)+ b^L_i q_i'(L)-\min\{q_i(L),q_i'(L)\})\;.
\end{align} 
We will bound each term $\nabla_H^L(\sigma,\tau,c)$ depending on whether $c\in L(v)$ or $c\notin L(v)$.

If $c\notin L(v)$, then it suffices to show that $\nabla_H^L(\sigma,\tau,c)\leq r \kappa^*$. 
Note that $a^L=b^L=0$, $q_i(L)=p_{a_i^L}$ and  $q'_i(L)=p_{b_i^L}$ for every $i\in [r]$. Let $\{c_i^L,d_i^L\}= \{a_i^L,b_i^L\}$ with $p_{c_i^L}\geq p_{d_i^L}$.  Using~\eqref{eq:list}, we obtain 
\begin{align*}
\nabla_H^L(\sigma,\tau,c) &\leq \sum_{i\in [r]} (a^L_i p_{a_i^L}+ b^L_i p_{b_i^L}-\min\{p_{a_i^L},p_{b_i^L}\}) \\
&= \sum_{i\in [r]} c^L_i p_{c_i^L}+ (d^L_i-1) p_{d_i^L}
\leq \frac{4}{3}r < r \kappa^*\;,
\end{align*} 
where we have used that $ip_i\leq 1$ and $(i-1)p_i\leq \frac{1}{3}$.

Now assume that $c\in L(v)$. We will compare these bounds with the ones we obtained in Section~\ref{sec:nablaH} by plugging the values of the $r$-configuration $(a^L_1,\dots,a^L_r;b^L_1,\dots,b^L_r)$. Observe that there are only two differences with respect to non-list colorings; first, $a^L$ and $b^L$ can be zero, and second, $a^L_{\max}$ and $b^L_{\max}$ can be zero. Recall that $p_0=0$.  It is important to stress that, since $c\in L(v)$, $a^L_{\max}=0$ implies $a^L=0$, and similarly for $b^L$. Therefore, the only difference between~\eqref{eq:list} and~\eqref{eq:bound nabla_H}, are the cases where either $a^L=0$ or $b^L=0$. If $a^L=a^L_{\max}=0$, then the total contribution of this part is zero and analogously for $b^L$. 
Therefore, the only interesting case is when $a^L=0$ and $a^L_{\max}\neq 0$; in this case $r\geq 2$. 
Since $a^L=0$ and $c\in L(v)$, there exists $j\in [r]$ such that $a^L_{j}=0$.  Consider the $(r-1)$-configuration 
\begin{align}\label{eq:reduced}
(a^L_1,\dots,a^L_{j-1},a^L_{j+1} ,\dots, a^L_r;b^L_1,\dots,b^L_{j-1},b^L_{j+1} ,\dots, b^L_r)\;.
\end{align}
Let $a=1+ \sum_{i\neq j}a^L_i$, $b=1+ \sum_{i\neq j}b^L_i$. Let $b_{\max}$ be the maximum of the $b^L_i$ with $i\neq j$ and note that $b_{\max}\leq b^L_{\max}$. Recall that $(\mathbf{p^*},\frac{11}{6})$ is an optimal solution of $(P)$ and $(\mathbf{p^*},\frac{161}{88})$ is an optimal solution of $(P^*)$ . If $r\geq 4$, then the $r$-configuration~\eqref{eq:reduced} for $c$ is non-extremal and
\begin{align*}
\nabla_H(\sigma,\tau,c)&\leq (a-a^L_{\max}-1) p_{a}+ (b-b_{\max}-1) p_{b} + \sum_{i\neq j} a_i^Lq_{i}+b_i^L q'_{i}-\min\{q_{i},q'_{i}\} \\
&\leq \frac{161}{88} (r-1)-1\;.
\end{align*}
If $1\leq r\leq 3$, then~\eqref{eq:reduced} can be extremal and
\begin{align*}
\nabla_H(\sigma,\tau,c)&\leq (a-a^L_{\max}-1) p_{a}+ (b-b_{\max}-1) p_{b} + \sum_{i\neq j} a_i^Lq_{i}+b_i^L q'_{i}-\min\{q_{i},q'_{i}\} \\
&\leq \frac{11}{6}(r-1)-1 \leq \frac{161}{88} (r-1)- \frac{131}{132}\;.
\end{align*}
For $i=i_a^L$ we have $q_{i}= p_{a_{\max}}-p_a$ and $q_{i}(L)= p_{a_{\max}}$. Moreover, we have $q_{j}'= 0$ and $q_{j}'(L)\leq p_{b_j^L}$. We may assume that $b_{\max}^L\neq 0$, as otherwise we have $b=b_{\max}=0$ and the contribution of this part is zero, as before. Using these bounds and~\eqref{eq:list}, we obtain that for any such $c\in [k]$
\begin{align*}
\nabla^L_H(\sigma,\tau,c) &\leq  (b^L-b^L_{\max}-1) p_{b^L} + \sum_{i\in [r]} (a_i^Lq_{i}(L)+b_i^L q'_{i}(L)-\min\{q_{i}(L),q'_{i}(L)\})\\
&\leq \nabla_H(\sigma,\tau,c) - (a-2a^L_{\max}-1) p_{a}+(b^L-b^L_{\max}-1) p_{b^L} +b^L_j p_{b^L_j}\\
&\leq \nabla_H(\sigma,\tau,c) + (a-1) p_{a}+(b^L-2) p_{b^L} +b^L_j p_{b^L_j}\\
&\leq \nabla_H(\sigma,\tau,c) + \frac{19}{12}\\
&\leq \frac{161}{88}\cdot r-1\;,
\end{align*}
where we have used that $b^L_{\max}\geq 1$, $a\geq a_{\max}^L+1$,  $ip_i\leq 1$, $(i-1)p_i\leq \frac{1}{3}$ and $(i-2)p_i\leq \frac{1}{4}$. Thus, Lemma~\ref{lem:improvement} also holds for $\nabla^L_H$.

Lemma~\ref{lem:complement} holds for $\nabla_B^L$ as well, since all the negative contributions on the bound are given by Kempe components $S=S_\sigma(u,c)$ of size $1$, which are always flippable as $c\in L(u)$. The positive contributions of the Kempe components is still bounded by the same quantity since, in the worst case, they are all flippable. 

Using the same flip parameters and reasoning as in the proof of Theorem~\ref{thm:improvement}, it follows that for every $k$-list assignment $L$ with $k\geq (\frac{11}{6}-\eta)\Delta$, flip dynamics for $L$-colorings mixes in time $kn
\log{(4n)}$, concluding the proof of Theorem~\ref{thm:list_flip}.
The transfer result of Vigoda (Theorem~\ref{thm:trans}) can be directly adapted to list colorings, and thus, Theorem~\ref{thm:list_Glau} follows as a corollary of Theorem~\ref{thm:list_flip}.

\section{Conclusion and open problems}\label{sec:open}
The main conjecture in the area is still wide open. 
\begin{conjecture}\label{conj:folklore}
For $k \geq \Delta+2$, Glauber dynamics for $k$-colorings has mixing time $O(n \log n)$.
\end{conjecture}
We note that Jerrum's original argument showing that Glauber dynamics for $k$-colorings is rapidly mixing for $k >2\Delta$ extends not only to list colorings but to a generalization of list coloring called correspondence coloring as well.  Informally speaking, in correspondence coloring each vertex has a list of available colors and each edge has a matching between the lists of the endpoints determining the conflicts between colors (see~\cite{DP17} for the formal definition).  Unfortunately flip dynamics for correspondence coloring is not well defined. In light of this, we raise the following question. 
\begin{question}
Does there exists $\eta>0$ such that Glauber dynamics for $k$-correspondence colorings is rapidly mixing provided that $k\geq (2-\eta) \Delta$?  
\end{question}
\noindent
Studying the correspondence coloring version of a problem can be illustrative, for instance correspondence coloring was introduced by Dvo\v{r}\'{a}k and the third author~\cite{DP17} to answer a long-standing question of Borodin concerning list coloring planar graphs without cycles of certain lengths.    
We believe that understanding the sampling of correspondence colorings with less than $2\Delta$ colors could shed light on how to tackle Conjecture~\ref{conj:folklore} for ordinary colorings.

Our primary objective in this paper was to show that there exists $\eta >0$ such that Glauber dynamics for $k$-colorings is rapidly mixing for $k \geq \left(\frac{11}{6}-\eta\right) \Delta$. Our approach could be refined in different ways in hope of obtaining a larger value of $\eta$. For instance, there are several parts of the proof where worst-case bounds are used (e.g. Lemma~\ref{lem:complement}) and a more careful analysis could yield to an improvement of $\eta$. A more challenging approach is to extend our notion of extremal configuration to include other configurations that are nearly extremal, such as $(1;1),(3;1),(4;1),(5;1)$ and $(2,2;1,1)$. Although, this would complicate the analysis of $\nabla_B$, it would likely lead to a larger $\eta$.

As a final remark, in contrast to  stopping-time-based metrics, our ``extremal'' metric only involves the study of one step of the chain and it is fairly easy to analyze. We believe this approach can have fruitful applications for bounding the mixing time of other Markov chains.

\end{document}